\documentclass[conference]{IEEEtran}
\ifCLASSINFOpdf
\else
\fi
\usepackage{amssymb}
\usepackage{amsmath}
\newtheorem{Approximation Lemma}{Approximation Lemma}
\newtheorem{Robustification Lemma}{Robustification Lemma}
\newtheorem{Covering Lemma}{Covering Lemma}
\newtheorem{Theorem}{Theorem}

\newtheorem{Lemma}{Lemma}

\newtheorem{Remark}{Remark}

\newtheorem{Compression Lemma}{Compression Lemma}

\def \exp{\text{exp}}

\begin{document}
\title{Capacities of classical compound quantum wiretap
and classical quantum compound wiretap channels}

\author{\IEEEauthorblockN{Minglai Cai}
\IEEEauthorblockA{Department of Mathematics\\
University of Bielefeld\\
Bielefeld, Germany\\
Email: mlcai@math.uni-bielefeld.de}
\and
\IEEEauthorblockN{Ning Cai}
\IEEEauthorblockA{The State Key Laboratory of \\
Integrated
Services Networks\\
University of Xidian\\
Xian, China\\
Email: caining@mail.xidian.edu.cn}
\and
\IEEEauthorblockN{Christian Deppe}
\IEEEauthorblockA{Department of Mathematics\\
University of Bielefeld\\
Bielefeld, Germany\\
Email: cdeppe@math.uni-bielefeld.de}}


%Telephone: (800) 555--1212, Fax: (888) 555--1212}


\maketitle

\thispagestyle{plain}\begin{abstract}
We determine the capacity of the classical compound quantum wiretapper
channel with channel state information at the transmitter. Moreover
we derive a lower bound on the capacity of this channel
without channel state information and determine the capacity of the classical
quantum compound wiretap channel with channel state information
 at the transmitter.
\end{abstract}

\section{Introduction}
The compound channel models transmission over a channel
that may take a number of states, its capacity was determined by
\cite{Bl/Br/Th}. A compound channel with an eavesdropper
is called a compound wiretap channel. It is defined as a family of pairs of
channels $\{(W_t,V_t) :t=1,\cdots,T\}$  with common input alphabet
and possibly different output alphabets, connecting a sender with
two receivers, one legal and one wiretapper, where $t$ is called a
state of the channel pair $(W_t, V_t)$. The legitimate receiver
accesses the output of the first channel $W_t$  in the pair
$(W_t,V_t)$, and the wiretapper observes the output of the second
part $V_t$ in the pair $(W_t,V_t)$, respectively, when a state $t$
governs the channel. A code for the channel conveys information to
the legal receiver such that the wiretapper knows nothing about the
transmitted information. This is a generalization of Wyner's wiretap
channel \cite{Wyn} to the case of multiple channel states.

We will be dealing with two communication scenarios. In the first
one only the transmitter is informed about the index $t$ (channel
state information (CSI) at the transmitter), while in the second,
the legitimate users have no information about that index at all (no
CSI).

The compound wiretap channels were recently introduced  in
\cite{Li/Kr/Po/Sh}. A upper bound on the capacity under the
condition that the average error goes to zero and the sender has no
knowledge about CSI is obtained. The result of \cite{Li/Kr/Po/Sh}
was improved in \cite{Bj/Bo/So} by
 using the stronger condition that the maximal error should go to
zero. Furthermore, the secrecy capacity for the case with CSI was
calculated.\vspace{0.15cm}

This paper is organized as follows.

In Section \ref{seccl} we present some known results for classical
 compound wiretap channel which we will use for our result's
 proof.

In Section \ref{secccqw} we derive the capacity of the classical
compound quantum wiretap channel with CSI and give a lower bound
of the capacity without CSI. In this channel model the wiretapper
uses classical quantum channels.

In Section \ref{seqcqw} we derive  the capacity of the classical
quantum compound wiretap channel with CSI. In this model both the
receiver and the wiretapper use classical quantum channels, and the
set of the states can be both finite or infinite. Here we will use an
idea which is similar to the one used in
\cite{Ah/Ca}.

\section{Classical Compound Wiretap Channels}\label{seccl}
Let $A$,$B$, and $C$ be finite sets, $P(A)$, $P(B)$, and $P(C)$ be the
sets of probability distributions on $A$, $B$ and $C$, respectively.
Let $\theta$ := $\{1,\cdots,T\}$. For every $t \in \theta$ let
$W_{t}$ be a channel $A \rightarrow P(B)$ and $V_{t}$ be a channel
$A \rightarrow P(C)$. We call $(V_t,W_t)_{t \in \theta}$ a compound
wiretap channel. $W_t^{ n}$ and $V_t^{ n}$ stand for the $n$-th
memoryless extensions of  stochastic matrices $W_t$ and $V_t$. 

Here the first family represents the communication link to the legitimate
receiver while the output of the latter is under control of the
wiretapper.

Let $X$ be a discrete random variable on a finite set $\{x_
1,\cdots,x_n\}$, with probability distribution function $p_i :=
Pr(x_i)$ for $i = 1,\cdots,n$, then the  Shannon entropy is
defined as
\[H(X):= \sum_{i=1}^{n} p_i \log p_i\text{ .}\] 
Let $X$ be a discrete random variable on a finite set $\mathfrak{X}$
with probability distribution function $P_X$, let $Y$ be a discrete
random variable on a finite set $\mathfrak{Y}$ with probability
distribution function $P_Y$, and let $P_{XY}$ be their joint
probability distribution, then the mutual information between $X$
and $Y$ is defined as
\[I(X,Y) := \sum_{x \in \mathfrak{X}, y\in \mathfrak{Y}}
P_{XY}(x,y) \log \frac{P_{XY}(x,y)}{P_X(x)P_Y(y)} \text{ .}\] 
Let $N(x|x^n)$ be
the number of occurrences of the symbol $x$ in the sequence $X^n$.
For a probability distribution $P \in P(A)$ and $\delta \geq 0$
 let typical sequences   and  conditional typical sequence be defined as :
\[ \mathcal{T}^n_P := \{x^n \in A^n : N(x|x^n) = n P(x) \forall x \in A\} \text{ ,}\]
\begin{align*}&\mathcal{T}^n_{P,\delta}
:= \{x^n \in A^n \\
&: |N(x|x^n) - n P(x)|
 \leq \delta \sqrt{nP(x)(1-P(x))}
 \forall x \in A\} \text{ .}\end{align*}

An $(n, J_n)$ code for the compound wiretap channel $(V_t,W_t)_{t
\in \theta}$  consists of  stochastic encoders $\{E\}$ : $\{
1,\cdots ,J_n\} \rightarrow P(A^n)$ and
 a collection of mutually disjoint sets $\left\{D_j
\subset B^n: j\in \{ 1,\cdots ,J_n\}\right\}$ (decoding sets).

 A non-negative number $R$ is an achievable secrecy rate for the
compound wiretap channel $(W_{t}, V_{t})$ in the case with CSI if
there is a collection of $(n, J_n)$ codes $(\{E_t: t \in \theta\},
\{D_j: j =1,\cdots,J_n\})$ such that
\[\liminf_{n \rightarrow \infty} \frac{1}{n}\log J_n \geq R \text{ ,}\]
\begin{equation} \label{b3}\lim_{n \rightarrow \infty} \max_{t \in
\theta} \max_{j\in \{ 1,\cdots ,J_n\}} \sum_{x^n \in A^n}
E_t(x^n|j)W_{t}^{n}(D_j^c|x^n)= 0 \text{ ,}\end{equation}
\begin{equation} \label{b4}\lim_{n \rightarrow \infty}
\max_{t \in \theta} I(J;Z_t^n) = 0\text{ ,}\end{equation} where $J$
is an uniformly distributed random variable with value in
$\{1,\cdots ,J_n\}$, and  $Z_t^n$ are the resulting random variables
at the output of wiretap channels  $V_t^{n}$. 

A non-negative number $R$ is an
 achievable secrecy rate for the
compound wiretap channel $(W_{t}, V_{t})$ in the case without CSI if
there is a collection of $(n, J_n)$ codes $(E, \{D_j: j=1,\cdots,J_n\})$ such that
\[\liminf_{n \rightarrow \infty} \frac{1}{n}\log J_n \geq R \text{ ,}\]
\begin{equation} \label{b3*}\lim_{n \rightarrow \infty} \max_{t \in
\theta} \max_{j\in \{ 1,\cdots ,J_n\}} \sum_{x^n \in A^n}
E(x^n|j)W_{t}^{n}(D_j^c|x^n)= 0 \text{ ,}\end{equation}
\begin{equation} \label{b4'}\lim_{n \rightarrow \infty}
\max_{t \in \theta} I(J;Z_t^n) = 0\text{ .}\end{equation}

\begin{Remark} A weaker and widely used security criterion is
obtained if we replace  (\ref{b4}), respectively (\ref{b4'}), with
$\lim_{n \rightarrow \infty}\max_{t \in \theta}\frac{1}{n}
 I(J;Z_t^n) = 0\text{ .}$\end{Remark}\vspace{0.15cm}

In case with CSI, let $p'_t (x^n):= \begin{cases} \frac{p_t^{
n}(x^n)}{p_t^{ n}
(\mathcal{T}^n_{p_t,\delta})} \text{ ,}& \text{if } x^n \in \mathcal{T}^n_{p_t,\delta}\\
0 \text{ ,}& \text{else} \end{cases}$\\
and $X^{(t)} := \{X_{j,l}^{(t)}\}_{j \in \{1, \cdots, J_n\}, l \in
\{1, \cdots, L_{n,t}\}}$  be a family of random matrices whose
entries are i.i.d. according to $p'_t$.

It was shown in \cite{Bj/Bo/So} that for any $\omega > 0$, if we set
\[J_n = \lfloor 2^{n(\min_{t \in \theta}(I(p_t,V_t)-\frac{1}{n}\log
L_{n,t})-\mu} \rfloor\text{ ,}\]
where $\mu$ is a positive constant  which does not depend
on $j$, $t$, and can be arbitrarily small when $\omega$ goes to $0$,
then there are such $\{D_j:j=1,\cdots,J_n\}$
that for all $t \in \theta$
\begin{align}&  Pr\left( \sum_{j=1}^{J_n} \frac{1}{J_{n}}
\sum_{l=1}^{L_{n,t}} \frac{1}{L_{n,t}} W_{t}^n
(D_j^c|X_{j,l}^{(t)}) > \sqrt{T}2^{-n\omega /2}\right) \notag\\
&\leq \sqrt{T}2^{-n\omega /2}\text{ .}\label{b6}\end{align} Since here only the
error of the legitimate receiver is analyzed, so for the result above
just the channels $V_t$, but not those of the wiretapper,
are regarded.

In view of  (\ref{b6}), one has (see \cite{Bj/Bo/So})\\
the largest achievable rate, called capacity,
of the compound wiretap channel with CSI at the
transmitter $C_{S,CSI}$, is given by
\begin{equation}\label{b1}
C_{S,CSI}= \min_{t \in \theta} \max_{V\rightarrow A \rightarrow
(BZ)_t}(I(V,B_t)-I(V,Z_t))\text{ ,}
\end{equation}
where $B_t$ are the resulting random variables at the output of
legal receiver channels. $Z_t$ are the resulting random
variables at the output of wiretap channels.

Analogously, in case without CSI, the idea is similar to the case
with CSI: Let $p' (x^n):= \begin{cases} \frac{p^{ n}(x^n)}{p^{ n}
(\mathcal{T}^n_{p,\delta})} & \text{if } x^n \in \mathcal{T}^n_{p,\delta}\\
0 & \text{else} \end{cases}$\\
 and $X^n := \{X_{j,l}\}_{j \in
\{1, \cdots, J_n\}, l \in \{1, \cdots, L_{n}\}}$  be a family of
random matrices whose components are i.i.d. according to $p'$.

For any $\omega > 0$, define
\[J_n = \lfloor 2^{n(\min_{t \in
\theta}(I(p_t,V_t)-\frac{1}{n}\log L_{n})-\mu} \rfloor\text{ ,}\]
where $\mu$ is a positive constant  which does not depend
on $j$, $t$, and can be arbitrarily small when $\omega$ goes to $0$,
then there are such $\{D_j:j=1,\cdots,J_n\}$
that for all $t \in \theta$
\begin{align}& Pr\left( \sum_{j=1}^{J_n} \frac{1}{J_{n}}
\sum_{l=1}^{L_{n}} \frac{1}{L_{n}} W_{t}^n (D_j(X)^c|X_{j,l}) >
\sqrt{T}2^{-n\omega /2}\right)\notag\\ & \leq \sqrt{T}2^{-n\omega
/2}\text{ .} \label{b6'}\end{align}
Using (\ref{b6'}) one can obtain (see \cite{Bj/Bo/So}) that
the secrecy capacity of the compound wiretap channel without CSI at
the transmitter $C_{S}$ is lower bounded as follows,
\begin{equation}
C_{S} \geq \max_{V\rightarrow A \rightarrow (BZ)_t}(\min_{t \in
\theta} I(V,B_t)-\max_{t \in \theta}I(V,Z_t))\text{ .} \label{b1'}
\end{equation}

\section{Classical Compound Quantum Wiretap Channels}\label{secccqw}
Let $A$ and $B$ be finite sets, and let $H$ be a finite-dimensional
complex Hilbert space. Let $P(A)$ and $P(B)$ be the sets of
probability distributions on $A$ and $B$ respectively, and
$\mathcal{S}(H)$ be the space of self-adjoint, positive-semidefinite
bounded linear operators with trace $1$ on $H$. Let $\theta :=
\{1,\cdots,T\}$ and for every $t \in \theta$ let $W_{t}$ be a
channel $A \rightarrow P(B)$ and $V_{t}$ be a classical quantum
channel, i.e., a map $A \rightarrow \mathcal{S}(H)$: $A \ni x
\rightarrow V_t(x) \in H$. We define $(V_t, W_t)_{t \in \theta}$ as
a classical compound quantum wiretap channel. Associate to $V_t$ is
the channel map on n-block $V_t^{\otimes n}$: $A^n \rightarrow
\mathcal{S}(H^{\otimes n})$ with $V_t^{\otimes n}(x^n) := V_t(x_1)
\otimes \cdots \otimes V_t(x_n)$.

For a state $\rho$, the von Neumann entropy  is defined as
\[S(\rho) := -\mathrm{tr} (\rho \log \rho)\text{ .}\]
Let $P$ be a probability distribution over a finite set $J$, and
$\Phi := \{\rho(x) : x\in J\}$  be a set of states labeled by
elements of $J$. Then the  Holevo $\chi$ quantity is
defined as
\[\chi(P,\Phi):= S\left(\sum_{x\in J} P(x)\rho(x)\right)-
\sum_{x\in J} P(x)S\left(\rho(x)\right)\text{ .}\]

An $(n, J_n)$ code for the  classical compound quantum wiretap
channel $(V_t,W_t)_{t \in \theta}$ consists of  stochastic encoders
$\{E\}$ : $\{ 1,\cdots ,J_n\} \rightarrow P(A^n)$ and
 a collection of mutually disjoint sets $\left\{D_j
\subset B^n: j\in \{ 1,\cdots ,J_n\}\right\}$ (decoding sets). 

A non-negative number $R$ is an achievable secrecy rate for the
classical compound quantum wiretap channel $(W_t,V_t)_{t \in
\theta}$ with  CSI if there is an  $(n, J_n)$ code $(\{E_t: t \in \theta\}, \{D_j: j=
1,\cdots, J_n \})$ such that
\[\liminf_{n \rightarrow \infty} \frac{1}{n}\log J_n \geq R\text{ ,}\]
\[\lim_{n \rightarrow \infty} \max_{t \in
\theta} \max_{j\in \{ 1,\cdots ,J_n\}} \sum_{x^n \in A^n}
E_t(x^n|j)W_{t}^{ n}(D_j^c|x^n)= 0\text{ ,}\]
\[\lim_{n \rightarrow \infty} \max_{t \in \theta} \chi(J;Z_t^{\otimes n}) = 0\text{ ,}\]
where $J$ is an uniformly distributed random variable with value in
$\{1,\cdots ,J_n\}$. $Z_t$ are the sets of states such that the
wiretapper will get. 

A non-negative number $R$ is an achievable
secrecy rate for the classical compound quantum wiretap channel
$(W_t,V_t)_{t \in \theta}$ without  CSI if there is an ($n$ $J_n$)
code $(E, \{D_j: j=1, \cdots, J_n\})$ such that
\[\liminf_{n \rightarrow \infty} \frac{1}{n}\log J_n \geq R\text{ ,}\]
\[\lim_{n \rightarrow \infty} \max_{t \in
\theta} \max_{j\in \{ 1,\cdots ,J_n\}} \sum_{x^n \in A^n}
E(x^n|j)W_{t}^n(D_j^c|x^n)= 0\text{ ,}\]
\[\lim_{n \rightarrow \infty} \max_{t \in \theta} \chi(J;Z_t^{\otimes n}) = 0\text{ .}\]

\begin{Theorem}\label{eq_1}
 The largest achievable rate (secrecy capacity) of the classical compound
quantum wiretap channel $(W_t,V_t)_{t \in \theta}$ in the case with
CSI $C_{S,CSI}$ at the transmitter is given by
\begin{equation}  C_{S,CSI} = \min_{t\in \theta} \max_{P \rightarrow A
\rightarrow B_t Z_t} (I(P,B_t)-\chi(P,Z_t)) \text{ .} \end{equation}
Respectively,   in the case without CSI, the secrecy capacity of the
classical compound quantum wiretap channel $(W_t,V_t)_{t \in
\theta}$ $C_{S}$ is lower bounded as follows
\begin{equation}  C_{S} \geq  \max_{P \rightarrow A
\rightarrow B_t Z_t} (\min_{t\in \theta}I(P,B_t)-\max_t \chi
(P,Z_t)) \text{ ,} \end{equation} where $B_t$ are the resulting
random variables at the output of legal receiver channels, and
$Z_t$ are the resulting random states at the output of wiretap
channels.
\end{Theorem}
\begin{proof}
\it 1) Lower bound \rm

Let $p'_t$, $X^{(t)}$, and $D_j$ be defined like in classical case.
Then (\ref{b6}) still holds since the sender transmits through a
classical channel to the legitimate receiver. We abbreviate
$\mathcal{X} := \{X^{(t)}: t\in \theta\}$.\vspace{0.15cm}

$\Big($Analogously, in the case without CSI, let $p'$ $X^{n}$ and $D_j$
be defined like in classical case, then (\ref{b6'}) still
holds.$\Big)$ \vspace{0.15cm}

For $\rho \in \mathcal{S}(H)$ and $\alpha > 0$ there exists an
orthogonal subspace projector $\Pi_{\rho ,\alpha}$ commuting with
$\rho ^{\otimes n}$ and satisfying
\begin{equation} \label{te1} \mathrm{tr} \left( \rho ^{\otimes n}
 \Pi_{\rho ,\alpha} \right) \geq 1-\frac{d}{\alpha ^2}\text{ ,}\end{equation}
\begin{equation} \label{te2} \mathrm{tr} \left( \Pi_{\rho ,\alpha} \right)
 \leq 2^{n S(\rho) + Kd\alpha \sqrt{n}}\text{ ,}\end{equation}
\begin{equation} \label{te3}  \Pi_{\rho ,\alpha} \cdot \rho ^{\otimes n} \cdot \Pi_{\rho ,\alpha} \leq
2^{ -nS(\rho) + Kd\alpha \sqrt{n}}\Pi_{\rho ,\alpha}\text{
,}\end{equation} where $a := \#\{A\}$,  and $K$ is a constant which is in polynomial order of $n$.\\
For $P\in P(A)$, $\alpha > 0$ and $x^n \in \mathcal{T}^n_P$  there
exists an orthogonal subspace projector $\Pi_{V,\alpha}(x^n)$
commuting with $V^{\otimes n}_{x^n}$ and satisfying:
\begin{equation} \label{te4} \mathrm{tr} \left( V^{\otimes n}(x^n) \Pi_{V,\alpha}(x^n) \right)
 \geq 1-\frac{ad}{\alpha ^2}\text{ ,}\end{equation}
\begin{equation} \label{te5} \mathrm{tr} \left( \Pi_{V,\alpha}(x^n) \right)
\leq 2^{n S(V|P) + Kad\alpha \sqrt{n}}\text { ,}\end{equation}
\begin{align}   &  \Pi_{V,\alpha}(x^n) \cdot V^{\otimes n}(x^n) \cdot \Pi_{V,\alpha}(x^n) \notag\\
&\leq 2^{ -nS(V|P) + Kad\alpha
\sqrt{n}}\Pi_{V,\alpha}(x^n)\text{ ,}\label{te6} \end{align}
\begin{equation} \label{te7}  \mathrm{tr} \left(  V^{\otimes n}(x^n) \cdot \Pi_{PV, \alpha \sqrt{a}} \right)
 \geq 1-\frac{ad}{\alpha ^2}\text{ ,}\end{equation}
where $a := \#\{A\}$, $d := \dim H$, and $K$ is a constant which is
in polynomial order of $n$ (see \cite{Wil}).\vspace{0.15cm}

 Let
\[Q_{t}(x^n) := \Pi_{PV_t, \alpha \sqrt{a}}\Pi_{V_t,\alpha}(x^n)
 \cdot V_{t}^{\otimes n}(x^n) \cdot \Pi_{V_t,\alpha}(x^n)\Pi_{PV_t, \alpha \sqrt{a}}\]\\
where  $\alpha$ will be defined later.\vspace{0.15cm}

\begin{Lemma} [see \cite{Win}] \label{eq_4a}
Let $\rho$ be a state and $X$ be a positive operator with $X  \leq
id$ (the identity  matrix) and $1 - \mathrm{tr}(\rho X)  \leq
\lambda \leq1$. Then
\begin{equation} \| \rho -\sqrt{X}\rho \sqrt{X}\| \leq \sqrt{8\lambda}\text{ .}
\end{equation}\end{Lemma}\vspace{0.15cm}

With the Lemma \ref{eq_4a}, (\ref{te1}),  (\ref{te7}), and the fact
that $\Pi_{PV_t, \alpha \sqrt{a}}$ and $\Pi_{V_t,\alpha}(x^n)$ are
both projection matrices,  for any $t$ and $x^n$ it holds:
\begin{equation} \label{eq_4}\|Q_{t}(x^n)-V_{t}^{\otimes n}(x^n)\| \leq
\frac{\sqrt{8(ad+d)}}{\alpha} \text{ .}\end{equation}

We set $\Theta_t:=  \sum_{x^n \in \mathcal{T}^n_{p_t,\delta}}
{p'}_t^{n}(x^n) Q_{t}(x^n)$. For given $z^n$ and $t$, $\langle
z^n|\Theta_t|z^n\rangle$  is the expected  value of $\langle z^n|
Q_{t}(x^n) |z^n\rangle$ under the condition $x^n \in
\mathcal{T}^n_{p_t,\delta}$.\vspace{0.15cm}

\begin{Lemma} [see \cite{Ahl/Win}]
Let  $\mathcal{V}$ be a finite dimensional Hilbert space, $X_1,
\cdots ,X_L$ be a  sequence of  i.i.d. random variables with values
in $\mathcal{S}(\mathcal{V})$ such that $X_i \leq \mu \cdot
id_{\mathcal{V}}$ for all $i \in \{1, \cdots , L\}$, and $\epsilon \in
]0,1[$. Let $p$ be a probability distribution on $\{X_1, \cdots
,X_L\}$, $\rho= \sum_{i} p(X_i)X_i$ be the expected value of $X_i$,
and
 $\Pi_{\rho,\lambda}'$ be the projector
onto the subspace spanned by the eigenvectors of $\rho$ whose
corresponding eigenvalues are greater than $\frac{\lambda}{\dim
\mathcal{V}} $, then
\begin{align}&
  Pr \left( \lVert L^{-1}
\sum_{i=1}^{L}X_i-\Pi_{\rho,\lambda}'\cdot\rho\cdot\Pi_{\rho,\lambda}'\rVert >  \epsilon \right)  \notag\\
& \leq 2\cdot (\dim \mathcal{V}) \exp
 \left( -L\frac{\epsilon^2\lambda}{2\ln 2(\dim \mathcal{V})\mu} \right) \text{ .}\label{eq_5}\end{align}
\end{Lemma}\vspace{0.15cm}

Let $\mathcal{V}$ be the image of  $\Pi_{P, \alpha \sqrt{a}}$. By
(\ref{te2}), we have \[\dim \mathcal{V} \leq 2^{n S(P) +
Kd\alpha\sqrt{an}}\text{ .}\] Furthermore
\begin{align}
&Q_{t}(x^n)\notag\\
&= \Pi_{PV_t, \alpha \sqrt{a}}\Pi_{V_t,\alpha}(x^n)
 \cdot V_{t}^{\otimes n}(x^n) \cdot \Pi_{V_t,\alpha}(x^n)\Pi_{PV_t, \alpha
 \sqrt{a}}\notag\\
& \leq 2^{-n(S(V_t|P) + Kad\alpha \sqrt{n})}\Pi_{PV_t, \alpha
\sqrt{a}}\Pi_{V_t,\alpha}(x^n)\Pi_{PV_t, \alpha
 \sqrt{a}}\notag\\
&\leq 2^{-n \cdot S(V_t|P) + Kad\alpha \sqrt{n} } \cdot \Pi_{PV_t, \alpha \sqrt{a}}\notag\\
&\leq 2^{-n \cdot S(V_t|P) + Kad\alpha \sqrt{n}} \cdot
id_{\mathcal{V}}\text{ .}\label{eq_5a}
\end{align}
The first inequality follows from (\ref{te6}). The second
inequality holds because $\Pi_{V_t,\alpha}$ and
$\Pi_{PV_t, \alpha \sqrt{a}}$ are projection matrices.
The third inequality holds because
$\Pi_{PV_t, \alpha \sqrt{a}}$ is a projection matrix onto
$\mathcal{V}$.\vspace{0.15cm}

Thus, by (\ref{eq_5}) and  (\ref{eq_5a})
\begin{align*}&Pr\left( \lVert \sum_{l=1}^{L_{n,t}} \frac{1}{L_{n,t}} Q_{t}(X^{(t)}_{j,l}) -
\Pi_{\Theta_t,\lambda}'\Theta_t\Pi_{\Theta_t,\lambda}' \rVert >
\frac{1}{2} \epsilon  \right)\\
 &\leq 2\cdot 2^{n (S(P) +
Kd\alpha\sqrt{an})} \\
&\cdot \exp \left( -L_{n,t}\frac{\epsilon^2}{8\ln  2} \lambda
 \cdot 2^{n(S(V_t|P)-S(P)) + Kd\alpha\sqrt{n}(\sqrt{a}-1)}
 \right)\notag\\
&=2\cdot 2^{n (S(P) + Kd\alpha\sqrt{an})} \notag\\
&\cdot \exp \left( -L_{n,t}\frac{\epsilon^2}{8\ln  2} \lambda \cdot
2^{n (-\chi(P,Z_t)) + Kd\alpha\sqrt{n}(\sqrt{a}-1)} \right)\text{
.}\end{align*} the equality in the last line holds since
\begin{align*}&S(P)-S(V_t|P) \\
&= S\left(\sum_{j} P(j) \sum_{l}\frac{1}{L_{n,t}}V_t^{\otimes n}(X^{(t)}_{j,l})\right)\\
&-\sum_{j} P(j)S\left(\sum_{l}\frac{1}{L_{n,t}}V_t^{\otimes n}(X^{(t)}_{j,l})\right)\\
&=\chi(P,Z_t)\text{ .}\end{align*} Notice that $\lVert\Theta_t-
\Pi_{\Theta_t,\lambda}'\Theta_t\Pi_{\Theta_t,\lambda}'\rVert \leq
\lambda$. Let $\lambda := \frac{1}{2}\epsilon$ and $n$ large enough
then
\begin{align}&Pr\left( \lVert \sum_{l=1}^{L_{n,t}} \frac{1}{L_{n,t}} Q_{t}(X_{j,l}) -
\Theta_t \rVert > \epsilon \right)\notag\\
& \leq2\cdot 2^{n (S(P) + Kd\alpha\sqrt{an})} \notag\\
&\cdot \exp \left( -L_{n,t}\frac{\epsilon^3}{16\ln 2 }\cdot 2^{n
(-\chi(P,Z_t)) + Kd\alpha\sqrt{n}(\sqrt{a}-1)} \right)\notag\\
& \leq \exp\left(-L_{n,t}\cdot2^{-n(\chi(P,Z_t)+\zeta)}\right)\text{
,}\label{eq_5b}\end{align} where $\zeta$ is  some suitable positive
constant, which does not depend on $j$, $t$, and can be
 arbitrarily small when $\epsilon$
goes to $0$.\vspace{0.15cm}

Let $L_{n,t} = 2^{n(\chi(P,Z_t)+2\zeta)}$ and $n$ be large enough,
then by (\ref{eq_5b}) for all $j$ it holds
\begin{equation}\label{eq_6b} Pr \left( \lVert \sum_{l=1}^{L_{n,t}}
 \frac{1}{L_{n,t}} Q_{t}(X^{(t)}_{j,l}) -\Theta_t
\rVert >  \epsilon  \right) \leq \exp(-2^{n \zeta })\end{equation} and
\begin{align} &Pr \left( \lVert \sum_{l=1}^{L_{n,t}} \frac{1}{L_{n,t}} Q_{t}(X^{(t)}_{j,l}) -\Theta_t
\rVert >  \epsilon  \text{ }\forall t\right)\notag\\
&  =1-Pr \left(\bigcup_{t}\{ \lVert \sum_{l=1}^{L_{n,t}}
\frac{1}{L_{n,t}} Q_{t}(X^{(t)}_{j,l})
-\Theta_t \rVert >  \epsilon  \} \right)\notag\\
 & \geq 1- T\exp(-2^{n \zeta })\notag\\
& \geq 1-2^{-n \upsilon} \text{ ,}\label{eq_6}
\end{align} where $\upsilon$ is some positive suitable constant
 which does not depend on $j$ and $t$.\vspace{0.15cm}

$\Big($Analogously,  in the case without CSI, let $L_{n} =  2^{n\max_{t}
(\chi(P,Z_t)+\delta)} $ and $n$ be large enough, then we can find
some positive constant $\upsilon$ so that
\begin{equation}\label{eq_6c} Pr \left( \lVert \sum_{l=1}^{L_{n}} \frac{1}{L_{n}} Q_{t}(X^{(t)}_{j,l}) -\Theta_t
\rVert > \epsilon \text{ }\forall t\right) \geq 1-2^{-n \upsilon}
\end{equation}
for all $j$.$\Big)$\vspace{0.15cm}

\begin{Remark} Since $\exp(-2^{n \zeta })$ converges to zero double exponentially faster,
 the inequality (\ref{eq_6}) remains true even if $T$ depends on $n$ and
is exponentially large over $n$,  i.e., we can still achieve
exponentially small error.
\end{Remark}\vspace{0.15cm}

From (\ref{b6}) and (\ref{eq_6}), it follows: For any $\epsilon >0$, if
$n$ is large enough then the event
\begin{align*}\left(\bigcap_{t} \left\{ \sum_{j=1}^{J_n} \frac{1}{J_{n}} \sum_{l=1}^{L_{n,t}}
\frac{1}{L_{n,t}} W_{t}^n (D_j^c(\mathcal{X})|X_{j,l}^{(t)}) \leq
\epsilon\right\}\right)\\
\cap \left(\bigcap_{j} \left\{  \lVert \sum_{l=1}^{L_{n,t}}
\frac{1}{L_{n,t}} Q_{t}(X_{j,l}^{(t)}) -\Theta_t \rVert \leq \epsilon
\text{ } \forall t \right\}\right)\end{align*} has a positive
probability.  This means that we can find a realization $x_{j,l}^{(t)}$ of
$X_{j,l}^{(t)}$
 with a positive probability such that for all $t \in \theta$,   we
have
\[ \sum_{j=1}^{J_n} \frac{1}{J_{n}}
\sum_{l=1}^{L_{n,t}} \frac{1}{L_{n,t}} W_{t}^n
(D_j^c|x_{j,l}^{(t)})\leq \epsilon\text{ ,}\]  and
\[ \lVert \sum_{l=1}^{L_{n,t}} \frac{1}{L_{n,t}} Q_{t}(x_{j,l}^{(t)}) -\Theta_t
\rVert  \leq \epsilon \text{ }\forall j\text{ .}\]

For any $\gamma
>0$ let \[R := \min_{t\in \theta} \max_{P \rightarrow A \rightarrow
B_t Z_t} \left(I(P,B_t)-\chi(P,Z_t)\right)+\gamma \text{ ,}\]  then
we have
\begin{equation}\label{eq_7}\liminf_{n \rightarrow \infty} \frac{1}{n}\log J_n \geq R \text{ ,}\end{equation}
\begin{equation}\lim_{n \rightarrow \infty} \max_{t \in
\theta} \max_{j\in \{ 1,\cdots ,J_n\}} \sum_{x^n \in A^n}
E_t(x^n|j)W_{t}^{ n}(D_j^c|x^n)= 0\text{ ,}\end{equation} where
$E_t$ is the random output of $(X_{j,l}^{(t)})_l$.

Choose a  sufficiently large but fixed $\alpha$ in (\ref{eq_4}) so
that for all $j$ it holds $\lVert V_t^{\otimes n}(x_{j,l}^{(t)}) -
 Q_{t}(x_{j,l}^{(t)})\rVert < \epsilon\text{ .}$
In this case, for any given $j' \in \{1, \cdots, J_n\}$ we have
\begin{align}
&\lVert \sum_{l=1}^{L_{n,t}} \frac{1}{L_{n,t}} V_t^{\otimes
n}(x_{j',l}^{(t)}) -
\Theta_t \rVert \notag\\
&\leq \lVert \sum_{l=1}^{L_{n,t}} \frac{1}{L_{n,t}} V_t^{\otimes
n}(x_{j',l}^{(t)}) - \sum_{l=1}^{L_{n,t}} \frac{1}{L_{n,t}}
Q_{t}(x_{j',l}^{(t)})\rVert \notag\\
&+ \lVert \sum_{l=1}^{L_{n,t}} \frac{1}{L_{n,t}} Q_{t}(x_{j',l}^{(t)}) - \Theta_t \rVert \notag\\
&\leq  \sum_{l=1}^{L_{n,t}} \frac{1}{L_{n,t}} \lVert V_t^{\otimes
n}(x_{j',l}^{(t)}) - Q_{t}(x_{j',l}^{(t)})\rVert \notag\\
&+ \lVert \sum_{l=1}^{L_{n,t}^{(t)}} \frac{1}{L_{n,t}} Q_{t}(x_{j',l}^{(t)}) - \Theta_t \rVert \notag\\
 &\leq 2\epsilon\label{eq_8}
\end{align}
and $\| \mathbb{E}_j \sum_{l=1}^{L_{n,t}} \frac{1}{L_{n,t}}
V_t^{\otimes n}(x_{j,l}^{(t)})-\Theta_t \| \leq \epsilon$ for any
probability distribution uniformly distributed on
$\{1,\cdots,J_n\}$.\vspace{0.15cm}

\begin{Lemma}[Fannes inequality \cite{Win}]\label{eq_9}
Let $\mathfrak{X}$ and $\mathfrak{Y}$ be two states in a
$d$-dimensional complex Hilbert space and
$\|\mathfrak{X}-\mathfrak{Y}\| \leq \mu < \frac{1}{e}$, then
\begin{equation} |S(\mathfrak{X})-S(\mathfrak{Y})| \leq \mu \log d
- \mu \log \mu \text{ .}\end{equation}\end{Lemma}\vspace{0.15cm}

 If $J$ is a
probability distribution uniformly distributed on
$\{1,\cdots,J_n\}$, then from the inequality (\ref{eq_8}) and Lemma
\ref{eq_9} we have
\begin{align}& \chi(J;Z_t^{\otimes n}) \notag\\
&=S\left(\mathbb{E}_j \sum_{l=1}^{L_{n,t}}
\frac{1}{L_{n,t}} V_t^{\otimes n}(x_{j,l}^{(t)})\right) \notag\\
&- \sum_{j=1}^{J_n}
J(j)S\left(\sum_{l=1}^{L_{n,t}}
 \frac{1}{L_{n,t}}V_t^{\otimes n}(x_{j,l}^{(t)})\right)\notag\\
  &\leq \vert  S\left(\mathbb{E}_j \sum_{l=1}^{L_{n,t}}
\frac{1}{L_{n,t}} V_t^{\otimes n}(x_{j,l}^{(t)})\right)-S\left( \Theta_t \right) \vert \notag\\
&+\vert  S(\Theta_t )- \sum_{j=1}^{J_n} J(j)S\left( \sum_{l=1}^{L_{n,t}}
\frac{1}{L_{n,t}} V_t^{\otimes n}(x_{j,l}^{(t)})\right)\vert \notag\\
&\leq \epsilon \log d - \epsilon \log \epsilon\notag\\
&+\vert \sum_{j=1}^{J_n} J(j) \left[S(\Theta_t )-S\left(
\sum_{l=1}^{L_{n,t}}\frac{1}{L_{n,t}} V_t^{\otimes n}(x_{j,l}^{(t)})
\right)\right] \vert \notag\\
 &\leq 3\epsilon \log d - \epsilon \log \epsilon -2\epsilon \log 2\epsilon\text{ .}\end{align}

We have
\begin{equation}\label{e10} \lim_{n \rightarrow \infty}
\max_{t \in \theta} \chi(J;Z_t^{\otimes n}) =
0\text{.}\end{equation}

 $\Big($Analogously,  in the case without CSI,  we
can find a realization $x_{j,l}^n$ of $X_{j,l}^{(t)}$
 with a positive probability such that:  For all $t \in \theta$,   we
have
\[ \sum_{j=1}^{J_n} \frac{1}{J_{n}}
\sum_{l=1}^{L_{n}} \frac{1}{L_{n}} W_{t}^n (D_j^c|x_{j,l})\leq
\epsilon\text{ ,}\]
\[ \lVert \sum_{l=1}^{L_{n}} \frac{1}{L_{n}} Q_{t}(x_{j,l}) -\Theta_t
\rVert  \leq \epsilon\text{ }\forall j \text{.}\] For any $\gamma
>0$ let
\[R := \max_{P
\rightarrow A \rightarrow B_t Z_t} \left( \min_{t\in \theta}
I(P,B_t)-\max_t \chi(P,Z_t)\right)+\gamma\text{ ,}\] then we have
\begin{equation}\label{eq_7'}\liminf_{n \rightarrow \infty} \frac{1}{n}\log J_n \geq R\text{ ,}\end{equation}
\begin{equation}\lim_{n \rightarrow \infty} \max_{t \in
\theta} \max_{j\in \{ 1,\cdots ,J_n\}} \sum_{x^n \in A^n}
E(x^n|j)W_{t}^{n}(D_j^c|x^n)= 0\text{ .}\end{equation}
 From $\lVert \sum_{l=1}^{L_{n}} \frac{1}{L_{n}} V_t^{\otimes n}(x_{j',l}) -
\Theta_t \rVert \rightarrow 0$ for $ n \rightarrow 0$ it follows
\begin{equation}\label{e10'} \lim_{n \rightarrow \infty}
\max_{t \in \theta} \chi(J;Z_t^{\otimes n}) =
0\text{ ,}\end{equation} for any
 probability distribution $J$ uniformly distributed on $\{1,\cdots,J_n\}$ in
the case without CSI.$\Big)$\vspace{0.15cm}

Combining (\ref{b6}) and (\ref{e10}) (respectively (\ref{e10'}))  we obtain
\[C_{S,CSI} \geq \min_{t \in \theta} \max_{V\rightarrow A \rightarrow
B_t Z_t}(I(V,B_t)-\chi (V,Z_t)) \text{ ,} \] respectively
\[ C_{S} \geq  \max_{P \rightarrow A
\rightarrow B_t Z_t} (\min_{t\in \theta}I(P,B_t)-\max_{t\in \theta}
\chi (P,Z_t))  \text{ .}\]\\[0.15cm]
\it 2) Upper bound for case with CSI \rm

Considering  $(\mathcal{C}_n)$ is a sequence of $(n,J_n)$ code such
that
\begin{equation}\label{eq_37}
 \sup_{t\in \theta}\frac{1}{J_n}\sum_{j=1}^{J_n}\sum_{x^n\in A^n}E(x^n|j)
 W_t^{ n}(D_j^{c}|x^n)=:  \epsilon_{1,n}\text{ ,}
\end{equation}
\begin{equation}
\sup_{t\in \theta} \chi(J;Z_t^{\otimes n}) =: \epsilon_{2,n}\text{ ,}
\end{equation} where $ \lim_{n\to\infty}\epsilon_{1,n}=0$ and
$\lim_{n\to\infty}\epsilon_{2,n}=0$, $J$ denotes the random variable
which is uniformly distributed on the message set $\{1,\ldots, J_n
\}$.

Let $C(V_{t},W_{t})$ denote
the secretey capacity of the wiretap channel $(V_{t},W_{t})$ in the sense
of \cite{Wil}. Choose $t'\in\theta$ such that  $C(V_{t'},W_{t'})=\min_{t\in\theta} C(V_{t},W_{t})$.

It is well-known, in  information theory, that even in
the case without wiretapper (we have only one classical
channel $W_{t'}$), the capacity cannot exceed $I(J;B_{t'})+\xi$ for
any constant $\xi > 0$. So the capacity of  a quantum wiretap
channel $(V_{t'},W_{t'})$ cannot be greater than
\begin{align*}&I(J;B_{t'})+
\xi \\
&\leq  \lim_{n\to\infty}[I(J;B_{t'})-\chi(J;Z_{t'}^{\otimes n})]+\xi
+  \epsilon_{2,n}\\
& \leq[I(J;B_{t'})-\chi(J;Z_{t'})]+\epsilon\end{align*} for any
$\epsilon > 0$.

Since we cannot exceed the secrecy capacity of the worst wiretap
channel, we have
\begin{equation}\label{b12}
C_{S,CSI} \leq \min_{t \in \theta} \max_{V\rightarrow A \rightarrow
B_t Z_t} (I(V,B_t)-\chi (V,Z_t))\text{ .}
\end{equation}\end{proof}\vspace{0.15cm}

\section{Classical Quantum Compound Wiretap Channel with CSI}\label{seqcqw}
Let  $H$ be a finite-dimensional complex Hilbert space. Let
$\mathcal{S}(H)$  be the space of self-adjoint,
positive-semidefinite bounded linear operators on  $H$ with trace
$1$. For every  $t \in \theta$ let $W_{t}$  respectively $V_{t}$ be
quantum channels, i.e.,  completely positive trace preserving maps
$\mathcal{S}(H) \rightarrow \mathcal{S}(H)$.

An $(n, J_{n}, \lambda )$  code for the classical quantum compound
wiretap channel $(W_t,V_t)_{t \in \theta}$ consists of a family of
vectors $w:=\{w(j): j=1,\cdots,J_n\} \subset S(H^{\otimes n})$ and
 a collection of positive semi-definite operators $\left\{D_j: j\in \{ 1,\cdots ,J_n\}\right\}
 \subset S(H^{\otimes n})$
which is a partition of the identity, i.e. $\sum_{j=1}^{J_n} D_j =
id_{H^{\otimes n}}$.

A non-negative number $R$ is an achievable secrecy rate for the
classical quantum compound wiretap channel $(W_t,V_t)_{t \in
\theta}$ with CSI if there is an $(n, J_{n}, \lambda )$ code
$(\big\{w_t:=\{w_t(j):j\}:t\big\},\{D_j:j\})$ such that
\[\liminf_{n \rightarrow \infty} \frac{1}{n}\log J_n \geq R\text{ ,}\]
\[\lim_{n \rightarrow \infty} \max_{t \in \theta}
 \frac{1}{J_n} \sum_{j=1}^{J_n} \mathrm{tr}\left(
W_t^{\otimes n}\left( w_t(j) \right)D_j\right)\geq 1 - \lambda\text{
,}\]
\[\lim_{n \rightarrow \infty} \max_{t \in \theta} \chi(J;Z_t^{\otimes n})= 0 \text{ ,}\]
where $J$ is a uniformly distributed random variable with value in
$\{1,\cdots ,J_n\}$, and $Z_t$ are the sets of states such that the
wiretapper will get.

\begin{Theorem}\label{e1}
The largest achievable rate (secrecy capacity) of the classical
quantum compound wiretap channel  in the case with CSI
 is given by
\begin{equation}\label{e1q}
C_{CSI} = \lim_{n \rightarrow \infty} \min_{t \in \theta} \max_{P,
w_t}\frac{1}{n}( \chi (P,B_t^{\otimes n})- \chi (P,Z_t^{\otimes
n})){ ,}\end{equation} where $B_t$ are the resulting random states
at the output of legal receiver channels, and $Z_t$ are the
resulting random states at the output of wiretap channels.
\end{Theorem}
\begin{proof}
Our idea is to send the information in two parts, firstly, we send
the state information with finite blocks of finite bits with a code
$C_1$ to the receiver, and then, depending on $t$, we send the message
with a code $C_2^{(t)}$ in the second part.\\[0.15cm]
\it 1) Sending channel state information with finite bits\rm

For the first  part, we don't require that the first part
should be secure against the wiretapper, since we assume that the
wiretapper already has the full knowledge of the CSI.

By ignoring the security against the wiretapper, we have only to
look at the compound channel $(W_t)_{t \in \theta}$.
 Let $W = (W_t)_t$ be an arbitrary compound
classical quantum channel. Then by \cite{Bj/Bo}, for each $\lambda
\in (0, 1)$ the $\lambda$-capacity $C(W, \lambda)$ equals
\begin{equation}
C(W,\lambda) = \inf_{t} \max_{ p } \chi (p,W_t)\text{ .}
\end{equation}
If $\min_{t} \max_{p} \chi (p,W_t) > 0$ holds, then the sender can
build a code $C_1$ such that the CSI can be sent to the legal
receiver with a block with length $l \leq \frac{\log T}{\min_{t}
\max_{ p } \chi (p,W_t)}$. We need to do nothing because
in this case the right hand side of (\ref{e1q}) is zero.

Let $\mathit{c} = 1-\lambda$, then for any required upper bound
$\delta = 2^{-c'}$, with given $c'
> 0$, the sender can repeat sending this block
$\log \mathit{c} \cdot c'$ times, and the legal receiver  simply
picks out the state that he receives most frequently to find out $t$
with a error probability $  \leq \delta $.

The first part is of length $l \cdot\log \mathit{c} \cdot  c' =
O(1)$, which is negligible compared to the second part.\\[0.15cm]
\it 2) Message transformation when both the sender and the
legal receiver Know CSI \rm 

If both the sender and the legal receiver
have the full knowledge of $t$, then we only have to look at the
single wiretap channel $(W_t,V_t)$.

In \cite{Ca/Wi/Ye} and \cite{De}, it is shown that there exists an
$(n, J_{n}, \lambda )$ code for the quantum wiretap channel $(W,V)$
with
\begin{equation}\log J_{n} = \max_{P, w}( \chi (P,B^{\otimes n})- \chi (P,Z^{\otimes n}))-\epsilon\text{ ,}
\end{equation} for any $\epsilon>0$, where $B$ is the resulting random variable at the output
of legal receiver's channel and $Z$ the output of the wiretap
channel.

When the sender and the legal receiver both know $t$, they can build
an $(n, J_{n,t}, \lambda )$ code $C_2^{(t)}$ where
\begin{equation}\log J_{n,t} =  \max_{P, w_t}( \chi (V,B_t^{\otimes n})
- \chi (V,Z_t^{\otimes n})) -\epsilon\text{ .}\end{equation}

Thus,
\begin{equation}\label{e3}
C_{CSI} \geq \lim_{n \rightarrow \infty} \min_{t \in \theta}
\max_{P, w_t}\frac{1}{n}( \chi (P,B_t^{\otimes n})- \chi
(P,Z_t^{\otimes n}))\text{ .}\end{equation}\vspace{0.15cm}

\begin{Remark}   For the construction of the second part of our code, we use
random coding and request that the randomization can be sent (see
\cite{Ca/Wi/Ye}). However, it is shown in \cite{Bj/Bo/So} that the
randomization could not always be sent if we require that we use one
unique code which is secure against the wiretapper and
 suitable for every channel state,  i.e., it does not depend on $t$.
This is not a counterexample to our results above, neither to the
construction of $C_1$ nor to the construction of $C_2^{(t)}$,
because of following facts.

The first  part of our code does not need to be secure. For our
second part, the legal transmitters can use the following strategy:
At first they bulid a code $C_1=(E,\{D_j:j=1,\cdots,J_n\})$ and a
code $C_2^{(t)}=(E^{(t)},\{D^{(t)}_j:j=1,\cdots,J_n\})$ for every $t
\in \theta$. If the sender wants to send the CSI $t'\in\theta$ and
the message $j$, he encodes $t'$ with $E$ and $j$ with $E^{(t')}$,
then he sends both parts together through the channel. After
receiving both parts, the legal receiver decodes the first part with
$\{D_j:j\}$, and chooses the right decoders
$\{D^{(t')}_j:j\}\in\left\{\{D^{(t)}_j:j\}:t \in \theta\right\}$ to
decode the second part. With this strategy, we can avoid using one
unique code which is suitable for every channel state.
\end{Remark} 
\it 3) Upper bound \rm

For any $\epsilon>0$ choose $t'\in\theta$ such that
$C(V_{t'},W_{t'})\leq\inf_{t\in\theta} C(V_{t},W_{t})+\epsilon$.

From \cite{Ca/Wi/Ye} and \cite{De}, we know that the capacity of the quantum
wiretap channel $(W_{t'},V_{t'})$ cannot be greater than
\[\lim_{n \rightarrow \infty}\max_{P, w_{t'}}\frac{1}{n}( \chi (P,B_{t'}^{\otimes
n})- \chi (P,Z_{t'}^{\otimes n}))\text{ .}\]
Since we cannot exceed
the capacity of the worst wiretap channel, we have
\begin{equation}\label{e4}
C_{CSI} \leq \lim_{n \rightarrow \infty}\min_{t \in \theta} \max_{P,
w_t}\frac{1}{n}( \chi (P,B_t^{\otimes n})- \chi (P,Z_t^{\otimes
n})){ .}\end{equation} This together with (\ref{e3}) completes the
proof of Theorem \ref{e1}. \end{proof}\vspace{0.15cm}

\begin{Remark} In \cite{Wa}, it is shown that if for a given $t$ and any $n \in \mathbb{N}$
 \[I(P,B_t^{\otimes
n}) \geq I(P,Z_t^{\otimes n})\] holds for all $P\in P(A)$ and
$\{w_t(j): j=1,\cdots,J_n\} \subset S(H^{\otimes n})$, then
\begin{align*}&\lim_{n \rightarrow \infty} \max_{P, w_t}\frac{1}{n}(
\chi (P,B_t^{\otimes
n})- \chi (P,Z_t^{\otimes n}))\\
& = \max_{P, w_t}(\chi (P,B_t)- \chi (P,Z_t))\text{ .} \end{align*}

Thus if for every $t\in\theta$ and $n \in \mathbb{N}$,
\[I(P,B_t^{\otimes n}) \geq I(P,Z_t^{\otimes
n})\] holds for all $P\in P(A)$ and $\{w_t(j): j=1,\cdots,J_n\}
\subset S(H^{\otimes n})$, we have
\[C_{CSI} = \min_{t \in \theta} \max_{P, w_t}( \chi (P,B_t)- \chi (P,Z_t))\text{ .}\]
\end{Remark}\vspace{0.2cm}

So far, we assumed that $|\theta|$, the number of the channels,
is $< \infty$. Now we look at the case where $|\theta|$ can be
arbitrary.\vspace{0.15cm}

\begin{Theorem} \label{e1*}
For an arbitrary set  $\theta$ we have
\begin{equation}C_{CSI}=\lim_{n \rightarrow \infty} \inf_{t \in \theta} \max_{P, w_t}\frac{1}{n}( \chi (P,B_t^{\otimes
n})- \chi (P,Z_t^{\otimes n}))\text{ .}\end{equation}
\end{Theorem}

\begin{proof}
Let $W: \mathcal{S}(H) \rightarrow  \mathcal{S}(H)$ be a linear map,
then let
\begin{equation}\|W \|_{\lozenge}:=\sup_{n\in \mathbb{N}}\max_{a\in S(\mathbb{C}^n
\otimes H), \|a\|_1=1}\| (id_n \otimes W)(a)\|_1\end{equation}
where $\|\cdot\|_1$ stands for the trace norm.

It is well known \cite{Pa} that this norm is multiplicative, i.e.
$\|W  \otimes W' \|_{\lozenge} = \|W\|_{\lozenge}\cdot \|W'
\|_{\lozenge}$.

A $\tau$-net in the space of the completely positive trace
preserving maps  is a finite set
$\left({W^{(k)}}\right)_{k=1}^{K}$
 with the property that for each $W$
there is at least one $k \in \{1, \cdots , K\}$
 with $ \| W-W^{(k)} \|_{\lozenge} < \tau$.
\begin{Lemma}[$\tau-$net \cite{Mi/Sch}]\label{e2}
For any $\tau \in (0,1]$ there is a $\tau$-net of quantum-channels
$\left(W_t^{(k)}\right)_{k=1}^{K}$ in  the space of the completely
positive trace preserving maps with $K \leq
(\frac{3}{\tau})^{2d^4}$, where $d = \dim H$.
\end{Lemma}\vspace{0.15cm}

If  $|\theta|$ is arbitrary, then for any $\xi >0$ let $\tau =
\frac{\xi}{-\log \xi}$. By Lemma \ref{e2} there exists a finite set
$\theta'$ with $|\theta'|\leq (\frac{3}{\tau})^{2d^4}$ and
$\tau$-nets $\left(W_{t'}\right)_{t' \in \theta'}$,
$\left(V_{t'}\right)_{t' \in \theta'}$ such that for every $
t\in\theta $ we can find a $t' \in \theta'$  with $\left\| W_t -
W_{t'}\right\|_{\lozenge}\leq \tau$ and $\left\| V_t -
V_{t'}\right\|_{\lozenge}\leq \tau$. For every $t'\in\theta'$ the
legal transmitters build a code $C_2^{(t')}
=\{w_{t'},\{D_{t',j}:j\}\}$. Since by \cite{Ca/Wi/Ye}, the error of
the code $C_2^{(t')}$ decreases exponentially to its length, we can
find an $N=O(-\log\xi)$ such that for all $t'\in \theta'$ it holds
\begin{equation}
 \frac{1}{J_N} \sum_{j=1}^{J_N} \mathrm{tr}\left(
W_{t'}^{\otimes N}\left( w_{t'}(j) \right)D_{t',j}\right)\geq 1 -
\lambda-\xi\text{ ,}\label{con1}\end{equation}
\begin{equation} \chi(J;Z_{t'}^{\otimes N})\leq
\xi\label{con2}\text { ,}\end{equation}

Then, if the sender obtains the state information ``$t$'' , he can
send with finite bits ``$t'$'' to the legal receiver in the first
part, and then they build a code $C_2^{(t')}$ that fulfills
(\ref{con1}) and (\ref{con2}) to transmit the message.

For every ${t'}$ and $j$ let $\psi_{t'}(j) \in H^{\otimes n} \otimes
H^{\otimes n}$ be an arbitrary purification of the state
$w_{t'}(j)$. Then we have
\begin{align*}&\mathrm{tr}\left[ \left(W_t^{\otimes N} - W_{t'}^{\otimes N}\right)(w_{t'}(j))\right]\\
&=\mathrm{tr} \left(\mathrm{tr}_{H^{\otimes N}} \left[id_H^{\otimes
N} \otimes (W_t^{\otimes N}-W_{t'}^{\otimes N})
\left( |\psi_{t'}(j)\rangle \langle \psi_{t'}(j) |\right)\right]\right)\\
& = \mathrm{tr} \left[id_H^{\otimes N} \otimes (W_t^{\otimes
n}-W_{t'}^{\otimes N})
\left( |\psi_{t'}(j)\rangle \langle \psi_{t'}(j) |\right)\right]\\
& = \left\|id_H^{\otimes N} \otimes (W_t^{\otimes N}-W_{t'}^{\otimes
N})\left( |\psi_{t'}(j)\rangle \langle \psi_{t'}(j) |\right)\right\|_1\\
& \leq \| W_t^{\otimes N}-W_{t'}^{\otimes N}\|_{\lozenge}
\cdot\left\|\left( |\psi_{t'}(j)\rangle \langle \psi_{t'}(j) |\right)\right\|_1\\
&\leq N\tau\text{ .}
\end{align*}
The first equality follows from the definition of purification. the
second equality follows from the definition of trace. The third
equality follows from the fact that $\|A\|_1=\mathrm{tr}(A)$ for any
self-adjoint, positive-semidefinite bounded linear operator $A$. The
first inequality follows by the definition of
$\|\cdot\|_{\lozenge}$. The second inequality follows from the facts
that $\|\left( |\psi_{t'}(j)\rangle \langle \psi_{t'}(j)
|\right)\|_1=1$ and $\left\| W_t ^{\otimes N}-W_{t'}^{\otimes
N}\right\|_{\lozenge}= \left\| \left(W_t-W_{t'}\right)^{\otimes
N}\right\|_{\lozenge} = N\cdot\left\| W_t -
W_{t'}\right\|_{\lozenge}$, since $\|\cdot\|_{\lozenge}$ is
multiplicative.\vspace{0.15cm}

It follows\begin{align}   &\frac{1}{J_N} \sum_{j=1}^{J_N}
\mathrm{tr}\left( W_t^{\otimes N}\left( w_{t'}(j)
\right)D_{t',j}\right) \notag\\
&-\frac{1}{J_N} \sum_{j=1}^{J_N}
\mathrm{tr}\left(
W_{t'}^{\otimes N}\left( w_{t'}(j) \right)D_{t',j}\right) \notag\\
&=\frac{1}{J_N} \sum_{j=1}^{J_N} \mathrm{tr}\left[ \left(
W_t^{\otimes N}
-W_{t'}^{\otimes N}\right)\left( w_{t'}(j) \right)D_{t',j} \right] \notag\\
&\leq \frac{1}{J_N} \sum_{j=1}^{J_N} \mathrm{tr}\left[ \left(
W_t^{\otimes N}
-W_{t'}^{\otimes N}\right)\left( w_{t'}(j) \right)\right] \notag\\
&\leq \frac{1}{J_N}J_N N\cdot\tau \notag\\
&= N\tau  \text{ .}\label{iea}\end{align}

$N\tau$ tends to zero when $\xi$ goes to zero, since $N=O(-\log\xi)$.

Let $J$ be a probability distribution  uniformly distributed on
$\{1,\cdots,J_N \}$, and
 $\{\rho(j) : j = 1,\cdots,J_n \}$  be a set of states
labeled by elements of $J$. By Lemma \ref{eq_9} we have
\begin{align}&\|\chi(J,V_t)-\chi(J,V_{t'})\| \notag\\
&\leq \Vert S\left(\sum_{j=1}^{J_N} J(j)
 V_t(\rho(j))\right)-S\left(
\sum_{j=1}^{J_N} J(j)V_{t'}(\rho(j))\right)\Vert \notag\\
&+\Vert \sum_{j=1}^{J_N} J(j)
 S\left(V_t(\rho(j))\right)
\sum_{j=1}^{J_N} J(j)S\left(V_{t'}(\rho(j))\right)\Vert \notag\\
 &\leq 2\tau\log d-2\tau\log\tau \text{ ,}\label{ieb}\end{align}
 since by $\left\| V_t -  V_{t'}\right\|_{\lozenge}\leq \tau$, it holds
 $\left\| V_t(\rho) -  V_{t'}(\rho)\right\|\leq \tau$ for all
 $\rho \in \mathcal{S}(H)$.\vspace{0.15cm}

By (\ref{iea}) and (\ref{ieb}) it holds
\[\max_{t}
 \frac{1}{J_N} \sum_{j=1}^{J_N} \mathrm{tr}\left(
W_t^{\otimes N}\left( w_{t'}(j) \right)D_{t',j}\right)\geq 1
-\lambda-\xi-N\tau\text{ ,}\]
\[ \chi(J;Z_t^{\otimes N})\leq \xi+2\tau\log d-2\tau\log\tau\text{ .}\]
Since $N\tau$ and $2\tau\log d$ both tend to zero when $\xi$ goes to
zero, we have
\[ \lim_{n \rightarrow \infty} \max_{t \in \theta}
 \frac{1}{J_n} \sum_{j=1}^{J_n} \mathrm{tr}\left(
W_t^{\otimes n}\left( w_{t'}(j) \right)D_{t',j}\right)\geq 1
-\lambda\text{ ,}\]
\[\lim_{n \rightarrow \infty}  \chi(J;Z_{t}^{\otimes n})= 0 \text{ .}\]

 The bits that the sender uses to
transform the CSI is large but constant, so it is still  negligible
compared to the second part.
We obtain
\begin{equation}C_{CSI}>\lim_{n \rightarrow \infty}\inf_{t \in \theta} \max_{P, w_t}\frac{1}{n}( \chi (P,B_t^{\otimes
n})- \chi (P,Z_t^{\otimes n}))\text{ .}
\end{equation}

The proof of the converse is similar to those given in 
the proof of Theorem \ref{e1}, where we consider a worst $t'$.
\end{proof}\vspace{0.15cm}

\begin{Remark} For Theorem \ref{e1} and Theorem \ref{e1*}, we have
 only required that the probability that the legal receiver does
not obtain the correct message tends to zero when the code length
goes to infinity. We have not specified how fast it should tends to
zero. If we analyze the relation between the error probability
 $\varepsilon$ and the code length, then we have the following facts.

In the case of finite $\theta$, let $\varepsilon_1$ denote the
probability that the legal receiver does not obtain the correct CSI,
and let $\varepsilon_2$ denote the probability that the legal
receiver, having CSI,  does not obtain the correct message. Since
the length of first part of the code is $l \cdot \log \mathit{c}
\cdot c' = O(\log \varepsilon_1)$, as we defined in Section
\ref{seqcqw}, we have $\varepsilon_1^{-1}$ is $O(\exp (l \cdot \log
\mathit{c} \cdot c'))=O(\exp(n))$,
where $n$ stands for  the length of first part.
 And for the second part of the code,
$\varepsilon_2$ decreased exponentially to the length of the second
part, as proven in \cite{Ca/Wi/Ye}. Thus, the error probability
 $\varepsilon = \max\{\varepsilon_1, \varepsilon_2\}$
decreases exponentially to the code length in the case of finite
$\theta$.

If $\theta$ is infinite, let $\varepsilon_1$ denote the
probability that the legal receiver does not obtain the correct CSI.
Then we have to build two  $\tau$-nets,
 each contains $O((\frac{-\log\varepsilon_1}{\varepsilon_1})^{-2d^4})$ channels.
If we want to send the CSI of these $\tau$-nets, $l$, as defined in
 Section \ref{seqcqw}, will be
$O(-2d^4\cdot\log(\varepsilon_1\log\varepsilon_1))$, this means here
$\varepsilon_1^{-1}$ will  be $O(\exp(\frac{n}{4d^4}))=O(\exp(n))$,
where $n$ stands for  the length of first part. So we can
still achieve that the error probability decreases exponentially to
the code length in case of infinite $\theta$.
\end{Remark}

\section*{Acknowledgment}
We thank Igor Bjelakovic and Holger Boche for useful discussions.
Support by the Bundesministerium f\"ur Bildung und Forschung (BMBF)
via grant 01BQ1052 is gratefully acknowledged.

\end{document}